\documentclass[prx,aps,twocolumn,groupaddress,floatfix]{revtex4-1}

\usepackage{amssymb,amsmath,amsthm,mathptm,graphicx,color,caption,times,xfrac,booktabs,mathtools,enumitem,xr}
\usepackage{subfigure}
\usepackage[unicode=true,bookmarks=true,bookmarksnumbered=false,bookmarksopen=false,breaklinks=false,pdfborder={0 0 1}, backref=false,colorlinks=true]{hyperref}

\def\bra#1{\langle{#1}\vert}
\def\ket#1{\vert{#1}\rangle}
\def\braket#1{\langle{#1}\rangle}

\def\BraVert{\egroup\,\mid\,\bgroup}

\def\tr#1{\mbox{tr}\left[{#1}\right]}
\newcommand{\ent}{S}
\newcommand{\q}{Q}
\newcommand{\s}{s}
\newcommand{\e}{r}
\newcommand{\se}{{sr}}

\definecolor{Blue}{rgb}{0,0,1}
\definecolor{Red}{rgb}{0,0,0}
\definecolor{Green}{rgb}{0,1,0}
\definecolor{Purp}{rgb}{.2,0,.2}
\definecolor{white}{rgb}{1,1,1}

\makeatletter
\theoremstyle{plain}
\newtheorem{thm}{\protect\theoremname}
\theoremstyle{plain}
\newtheorem{lem}[thm]{\protect\lemmaname}
\theoremstyle{plain}

\theoremstyle{remark}
\newtheorem*{rem*}{\protect\remarkname}
\theoremstyle{plain}

\theoremstyle{plain}

\theoremstyle{definition}

\providecommand{\propositionname}{Proposition}
\providecommand{\theoremname}{Theorem}
\providecommand{\lemmaname}{Lemma}
\providecommand{\remarkname}{Remark}
\providecommand{\conjecturename}{Conjecture}
\providecommand{\definitionname}{Definition}
\providecommand{\corollaryname}{Corollary}

\usepackage{bbm} %for the the mathbb {1} (identity matrix) \mathbbm{1}
\allowdisplaybreaks

\newcommand{\dif}{\textnormal{d}}
\newcommand{\expo}{\text{e}}
\newcommand{\kb}[2]{\ket{#1} \bra{#2}}
\newcommand{\ptr}[2]{\operatorname{tr}_{#1}
\left[ #2 \right]}
\newcommand{\normp}[2]{\left\| #2 \right\|_{#1}}
\DeclareMathOperator{\Prob}{Prob}

\begin{document}

\title{Emergence of a fluctuation relation for heat in nonequilibrium Landauer processes}

\author{Philip Taranto}
\email{philip.taranto@monash.edu}
\affiliation{School of Physics \& Astronomy, Monash University, Victoria 3800, Australia}

\author{Kavan Modi}
%\email{kavan.modi@monash.edu}
\affiliation{School of Physics \& Astronomy, Monash University, Victoria 3800, Australia}

\author{Felix A. Pollock}
%\email{felix.pollock@monash.edu}
\affiliation{School of Physics \& Astronomy, Monash University, Victoria 3800, Australia}

\date{\today}

\begin{abstract}
In a generalised framework for the Landauer erasure protocol, we study bounds on the heat dissipated in typical nonequilibrium quantum processes. In contrast to thermodynamic processes, quantum fluctuations are not suppressed in the nonequilibrium regime and cannot be ignored, making such processes difficult to understand and treat. Here, we derive an emergent fluctuation relation that virtually guarantees the average heat produced to be dissipated into the reservoir when either the system or reservoir is large (or both), or when the temperature is high. The implication of our result is that for nonequilibrium processes, heat fluctuations away from its average value are suppressed independently of the underlying dynamics exponentially quickly in the dimension of the larger subsystem and linearly in the inverse temperature. We achieve these results by generalising a concentration of measure relation for subsystem states to the case where the global state is mixed. %We also randomly sample dynamics to show that the fluctuation relation resolves to a tighter bound on the average heat than previously known bounds.
\end{abstract}
\maketitle
\section{Introduction} 

Landauer's principle \textbf{(LP)} provides the clearest evidence that ``information is physical'' by relating logically-irreversible computations to a necessary energy expenditure~\cite{Landauer1961}. The principle lies at the interface between information theory and thermodynamics: simultaneously offering deep consequences for the foundations of physics whilst positing the daunting technical challenge of managing heat dissipation in computers, whether they operate on classical or quantum logic. Indeed, although LP was initially postulated from classical thermodynamic considerations~\cite{Landauer1961,Bennett1982}, early research efforts aimed to either develop a microscopic, nonequilibrium version of the principle~\cite{Shizume1995, Piechocinska2000, Lutz2009, Sagawa2009} or extend it into the quantum domain~\cite{Plenio1999, Vedral2000, Janzig2000, Vaccaro2009, Hilt2011, Esposito2011, Barnett2013}. However, the microscopic versions often relied on specific models, and many quantum extensions assumed the principle to hold \emph{a priori}, before investigating its implications. Perhaps surprisingly, recent experiments demonstrate that LP applies to irreversible, nonequilibrium processes involving individual quantum systems~\cite{Berut2012, Orlov2012, Jun2014, Silva2014}. This evidence sparked a revival of interest in developing a rigorous formulation of LP in this regime, culminating in an equality form of LP derived by Reeb \& Wolf \textbf{(RW)} within a minimal framework~\cite{Reeb2014}.  

Despite the substantial body of work surrounding LP, little is known about the tightness of the bound at microscopic scales, nor about how Landauer heat can be tamed, \emph{i.e.}, how to minimise the heat required to process quantum information~\cite{Mohammady2016}. In the nonequilibrium setting of logical processes, where highly entangling operations are often utilised, the typical behaviour of the heat generated is not immediate, even on average -- thermodynamic intuition breaks down due to the strong interaction. On a technical level, minimising this heat is crucial for our ability to manipulate quantum systems to outperform their classical counterparts, as the quantum advantage often relies on coherent control that suffers from heat fluctuations. An approach to resolving these outstanding issues makes use of tools from nonequilibrium statistical physics~\cite{Esposito2009, Brandao2013, Tasaki2016, Goold2015, Lorenzo2015,Faist2015,Croucher2017}. In particular, RW~\cite{Reeb2014} and Goold, Paternostro \& Modi \textbf{(GPM)}~\cite{Goold2015} independently derived tighter bounds on heat (in the quantum setting) than that of Landauer. Although the RW correction only depends on the dimension of the reservoir and there is no strict hierarchy between the two novel bounds, the GPM bound tends to outperform the former in various regimes of interest; however, its calculation depends explicitly on details of the process and is therefore difficult to estimate in general~\cite{Goold2015}.

On a more fundamental level, LP is often cited as an equivalent formulation of the Second Law of Thermodynamics -- indeed it is when considering processes taking place near equilibrium. Investigations into statistical formulations of the Second Law which hold in the nonequilibrium regime have lead to the much lauded fluctuation relations~\cite{Jarzynski1996,Jarzynski1997,Crooks1998,Crooks1999}, which bound process-dependent quantities (such as work) to thermodynamic ones (such as the free energy). It is natural to ask whether one can develop similar process-independent bounds on heat by studying LP in the nonequilibrium regime. Such a formulation is indeed pertinent to the field of quantum thermodynamics, where processes inherently take place far from equilibrium due to the mesoscopic nature of the relevant subsystems~\cite{Goold2015Review, Vinjanampathy2015}. 

The aim of this Article is therefore to understand what universal properties typically emerge with respect to the heat generated in open quantum processes. We prove the emergence of a fluctuation relation for the heat dissipated in a Landauer process, stating that on average, heat is almost always dissipated into the environment. We analytically prove that this fluctuation relation arises exponentially quickly as the dimensions of either subsystem grows, and linearly in the inverse temperature. Our result extends the minimal framework for describing Landauer processes~\cite{Reeb2014} and is derived by examining fluctuations of the heat distribution~\cite{Goold2015}. We begin the Article by introducing the former and constructing the latter.

\section{Background} 

Surprisingly, until recently there was no consensus on how LP should be quantitatively expressed. This changed when RW formally derived a bound for the dissipated heat under a minimal set of assumptions~\cite{Reeb2014}: (i) the irreversible process involves a system $\s$ and a reservoir $\e$; (ii) the initial joint state is uncorrelated: $\rho_{\se} = \rho_\s \otimes \rho_\e$; (iii) the reservoir is initially in a thermal (Gibbs) state $\rho_\e := \expo^{- \beta H_\e}/Z$, where $\beta$ is the inverse temperature, $H_\e$ is the reservoir Hamiltonian, and $Z := \tr{\expo^{- \beta H_\e} }$ is the partition function; and (iv) the joint state evolves unitarily: $\rho_{\se}^\prime = U \rho_{\se} U^\dagger$. We call such processes \emph{Landauer processes}.

RW show that relaxing any one of the assumptions above can lead to violation of the bound
\begin{gather}\label{rwbound}
\beta \braket{\q} \ge \Delta \ent + R(\Delta \ent,d_\e) := \omega,
\end{gather}
where the inequality without $R(\Delta \ent,d_\e)$ is Landauer's bound. Here, the average heat dissipated into the reservoir is $\braket{\q} := \tr{H_\e (\rho_\e^\prime - \rho_\e)}$; the change in von Neumann entropy of the system is $\Delta S := \ent(\rho_\s) - \ent(\rho_\s^{\prime})$ with $\ent(\rho) := -\tr{\rho \log(\rho)}$; and  $R(\Delta \ent,d_\e) \ge 0$ is a correction term that tightens the Landauer bound for finite-sized reservoirs~\footnote{This Article only considers finite-dimensional systems.}.

Crucially, the above framework accounts for processes that will be employed by realistic quantum technologies: namely, nonequilibrium processes that lie outside the realm of traditional thermodynamics and are difficult to treat because of heat fluctuations that are not suppressed. In other words, the heat generated in a single run of the process can vary drastically from its average behaviour. The modern approach to describing such nonequilibrium processes employs fluctuation relations~\cite{Jarzynski1996, Jarzynski1997, Crooks1998, Crooks1999, Tasaki2000, Talkner2007, Talkner2009, Saira2012, Jarzynski2011, Campisi2011,Rastegin2013, Albash2013, Rastegin2014, Goold2014, Goold2015Review, Vinjanampathy2015, Deffner2016}. These relate thermodynamic quantities (\emph{e.g.}, free energy difference) to nonequilibrium quantities (\emph{e.g.}, work or heat), offering a promising route to understanding the thermodynamics of small systems whose relevant dynamics may occur on shorter timescales than equilibration. Recent work demonstrates that this formalism is applicable for the experimental exploration of quantum thermodynamics~\cite{Dorner2013, Mazzola2013, Batalhao2014, Goold2014Heat}; importantly including measuring the heat distribution of a quantum Landauer process~\cite{Silva2014}.

Applying such tools to the Landauer protocol, GPM developed a novel bound for the average heat~\cite{Goold2015}. By taking projective measurements of the reservoir energy, the complete distribution of the heat exchanged can be constructed: $P(\q) := \sum_{mn} P(E_m|E_n)P(E_n) \, \delta ( \q - (E_n - E_m))$, where $P(E_n) = \braket{E_n | \rho_\e | E_n}$ and $P(E_m|E_n) = \sum_l |\braket{E_n | A_l  |E_m}|^2$ are the initial and (conditional) final measurement probabilities respectively. The $A_{l=jk} = \sqrt{\lambda_j} \braket{k|U|j}$ are Kraus operators describing the local action of the evolution on the reservoir, with $\{ \lambda_j , \ket{j} \}$ the eigenvalues and eigenstates of $\rho_\s$ respectively. From this distribution, GPM show that the average exponentiated heat can be written as
\begin{gather}\label{gpmbound}
\Gamma := \braket{\expo^{-\beta \q}} = \tr{U^\dag \mathbbm{1}_\s \otimes \rho_\e U \rho_s \otimes \mathbbm{1}_\e},
\end{gather}
where $\braket{\expo^{-\beta \q}} = \int \dif  \q \, P( \q) \, \expo^{- \beta  \q}$. Invoking Jensen's inequality~\footnote{Jensen's inequality holds for any convex function $f$ of a random variable $X$: \unexpanded{$ \braket{ f(X) } \ge f(\braket{X})$}.}, the GPM bound is immediately derived
\begin{gather}\label{gpmbound2}
\beta \braket{\q} \ge -\ln(\Gamma) =: \gamma.
\end{gather}
Although the derivation above is reminiscent of the Jarzynski equality~\cite{Jarzynski1996, Jarzynski1997}, Eq.~\eqref{gpmbound} is not a true fluctuation relation since $\Gamma$ depends explicitly on the details of the process, \textit{i.e.}, $\Gamma = \Gamma (U, \rho_\s)$. However, we now show that a fluctuation relation emerges quickly for typical Landauer processes, meaning deviations of $\Gamma$ from its average value become suppressed \emph{independently of process details}. 

\section{Main Results}

\textbf{Emergence of fluctuation relation.---}Our main result shows that, for Haar-randomly sampled joint-space unitary interactions, a fluctuation relation for heat arises in the limit where the dimension of the system, reservoir, or both, becomes large; or when the temperature is high, \emph{i.e.}, $\Gamma \to 1$. In fact, as the dimensions of either $\s$ or $\e$ grow, the deviations of $\Gamma$ from unity are at least exponentially suppressed; in the high temperature limit this suppression is at least linear. First, we demonstrate the exponential scaling with dimension through the following Theorem:
\begin{thm} \label{LargeLimitTheorem}
When either the system or the reservoir dimensions are much larger than the other, \textnormal{i.e.},  $d_\s \ll d_\e$ or $d_\s \gg d_\e$, the deviations of $\Gamma$ from unity are at least exponentially suppressed in the dimension of the larger subsystem.
\end{thm}

Note first that we can write Eq.~\eqref{gpmbound} as $\Gamma = d_\s \tr{M_\s \rho_\s} = d_\e \tr{M_\e \rho_\e}$, where
\begin{gather}\label{opMs}
M_\s := \ptr{\e}{U^\dag \frac{\mathbbm{1}_\s}{d_\s} \otimes \rho_\e U},\quad
M_\e := \ptr{\s}{U \rho_\s \otimes \frac{\mathbbm{1}_\e}{d_\e} U^\dagger}.
\end{gather}
The following Lemma is a generalisation of standard concentration of measure results for quantum states (see, \textit{e.g.}, Ref.~\cite{Popescu2006}) to the case where the reduced density operators are generated from \emph{unitary orbits of mixed states}:

\begin{lem} \label{lem:levymixed}
For any $\sigma_\se = U \tau_\se U^\dagger$, where $\tau_\se$ is a fixed system-reservoir density operator and $U$ is a Haar-randomly sampled unitary operator
\begin{gather}  \label{LevyBound}
\Prob\left[ \normp{1}{\sigma_\s - \frac{\mathbbm{1}_\s}{d_\s}} \geq \sqrt{\frac{d_\s}{d_\e}} + \epsilon \right] \leq2 \exp{ \left( - \frac{d_\s d_\e \epsilon^2 }{16} \right)}.
\end{gather}
The same holds with system and reservoir labels swapped.
\end{lem}

Here, $\sigma_\s := \ptr{\e}{\sigma_\se}$ and the trace norm is defined for an operator $A$ as \unexpanded{$\normp{1}{A}:= \tr{\sqrt{A^\dagger A}}$}. Importantly, Lemma~\ref{lem:levymixed} bounds the trace distance of a reduced state from the maximally mixed state for Haar-randomly sampled joint interactions. While the bound in Eq.~\eqref{LevyBound} is the same as in the usual case of pure joint states, which follows from Levy's Lemma~\cite{LedouxBook}, the extension to mixed joint states is nontrivial, as the geometry of the corresponding space differs considerably. In fact, the following results do not hold for a na\"{i}ve application of the pure state result, because the trace distance from the identity of $\sigma_\s$, generated from a convex mixture $\sigma_\se$, cannot be directly upper bounded by an arbitrary component of the mixture. For a proof of this physically motivated application of Levy's Lemma, see Appendix~\ref{app:levyproof}~\footnote{Upon completion of the manuscript, we became aware that a similar result is used in relation to the decoupling protocol in Refs.~\cite{Hutter2012,Dupuis2014} based on Corollary 4.4.28 of~\cite{AndersonBook}.}.

\begin{proof} (Theorem~\ref{LargeLimitTheorem})
Consider the case where $d_\e \gg d_\s$.  For a Haar-randomly chosen unitary $U$, the state $M_\s$, defined in Eq.~\eqref{opMs}, is distributed exactly as $\sigma_\s$ in Lemma~\ref{lem:levymixed} with $\tau_\se= \left(\mathbbm{1}_\s/d_\s\right) \otimes\rho_\e$. Writing $\mu_\s :=  \normp{1}{M_\s - \mathbbm{1}_\s/d_\s} $, it follows immediately that: $\Prob \left[ {\mu_\s} \geq \sqrt{{d_\s}/{d_\e}} + \epsilon \right] \leq 2 \exp{ \left( - {d_\s d_\e \epsilon^2 }/{16} \right)}$. Choosing $\epsilon = 4\sqrt{(x d_\e  + \ln(2 d_\s))/(d_\s d_\e)}$ for some small $x > 1/d_\e$ gives
\begin{gather}\label{probmus}
\Prob\left[ {\mu_\s} \geq  \sqrt{\frac{d_\s}{d_\e}} +  4\sqrt{\frac{x d_\e  + \ln(2 d_\s)}{d_\s d_\e}}\right] \leq \frac{\exp{\left( - d_\e x \right)}}{d_\s}.
\end{gather}
As the reservoir dimension increases, independently of the inverse temperature $\beta$ and for a fixed $d_\s \geq 2$, the probability that ${\mu_\s}$ is greater than some vanishingly small quantity is exponentially diminishing in $d_\e$. Now, consider that $\mu_\s =\max_P \tr{P\left(M_\s - {\mathbbm{1}_\s}/{d_\s}\right)}$, where the maximisation is taken over all projection operators $P$. Using the fact that $\rho_s$ is a convex mixture of projectors, and multiplying $\mu_\s$ by $d_\s$, we have $d_\s \, \mu_\s \geq  d_\s  \tr{\rho_\s \left(M_\s - \frac{\mathbbm{1}_\s}{d_\s} \right)} =  \Gamma - 1$. By symmetry of the trace distance, we also have $- d_\s \, \mu_\s \leq \Gamma - 1$; thus we have upper bounded the fluctuations of $\Gamma$ about 1 by
\begin{gather}\label{projmax}
d_\s \mu_\s \geq | \Gamma - 1 | =: \mu.
\end{gather}
In summary, the magnitude of these fluctuations are upper bounded by a number that has high probability of being extremely small, as shown in Eq.~\eqref{probmus}. It follows that $\Gamma \to 1$ at least exponentially in the limit $d_\e \gg d_\s$. 
The Theorem can be proved for $d_\s \gg d_\e$ using a similar argument with the state ${M_\e}$, defined in Eq.~\eqref{opMs}.
\end{proof}

\begin{figure*}[ht]
\centering
\includegraphics[width=.99\linewidth]{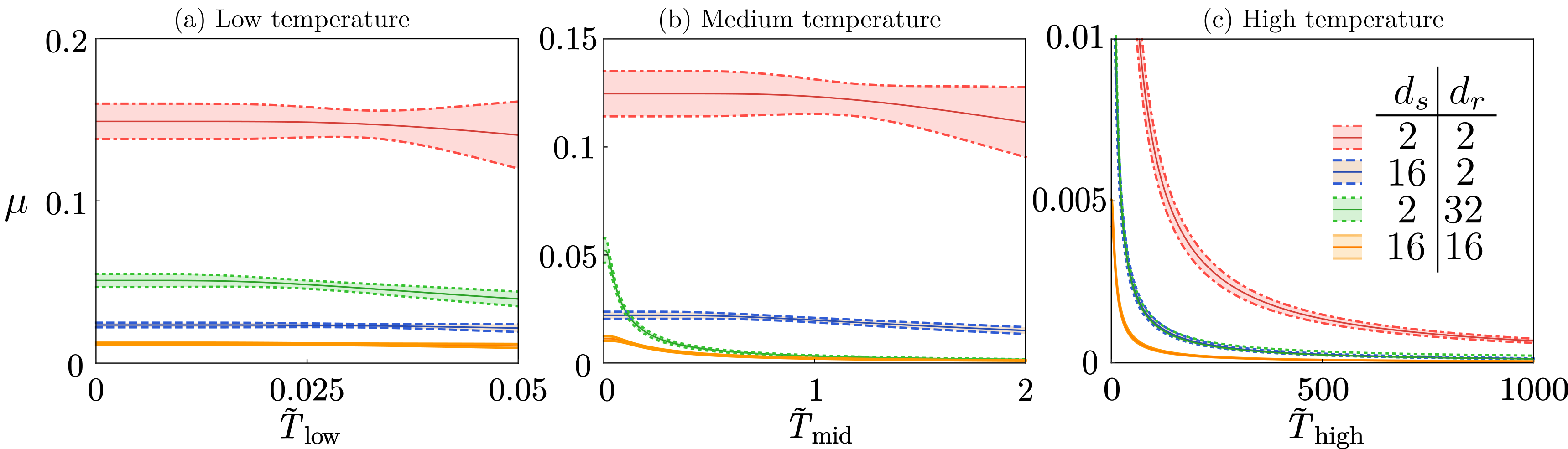}
\caption{(Color online). A fit of $\mu = |\Gamma - 1|$ to the function $a(1-\exp[-b/T])$ -- which we expect to approximate the size of deviations in the large $T$ limit -- for a variety of system and reservoir sizes in the {\bf (a)} low, {\bf (b)} intermediate and {\bf (c)} high temperature regimes. The central lines are the fitted values of $\mu$ and the shaded regions indicate the $\sim95\%$ confidence region for the fit. To calculate the fit, we sampled 1000 Haar-random unitary transformations applied to randomly chosen nonequilibrium system states for each panel in each of the following cases:  (i) $d_\s = d_\e = 2$ (red, dot-dashed); (ii) $d_\s = 16 \gg d_\e = 2$ (blue, dashed); (iii) $d_\s = 2 \ll d_\e = 32$ (green, dotted); and (iv) $d_\s = d_\e = 16$ (orange, solid). Panels (a) and (b) show that even at low or medium temperatures, when either subsystem dimension (or both) is large, $|\mu| \to 0$. Panel (c) shows that independently of the system and reservoir dimensions, $|\mu| \to 0$ in the high temperature limit. The only remaining case of interest is when the system and reservoir dimensions are both small and the temperature is low (see case (i) in panel (a)).} \label{TemperaturePlots}
\end{figure*}

The probabilistic statement we make in Theorem~\ref{LargeLimitTheorem} is based on the trace distance of the reduced post-dynamics state from the maximally mixed state. Lemma~\ref{lem:levymixed} states that as either dimension increases, typically, this distance is exponentially diminishing. This quantity upper bounds the absolute difference between $\Gamma$ and its mean value of 1, the implication being that the distribution of $\Gamma$ is sharply peaked. This further implies that the distribution of $\gamma = - \ln{(\Gamma)}$, which lower bounds the average heat, is sharply peaked around 0 in these asymptotic regimes.

Consider now the case where we have a large nonequilibrium system interacting with a large equilibrium reservoir, \emph{i.e.}, $d_\e \approx d_\s \gg 2$. From the above concentration of measure argument, it is not clear how $\Gamma$ behaves when $d_\s$ and $d_\e$ are comparable. With the following Theorem, we show that a fluctuation relation also emerges when the overall dimension becomes large:

\begin{thm}\label{HighTotDim}
When the system and reservoir dimensions are similar, we expect $\Gamma\rightarrow 1$ for large $d_{\se}=d_\s d_\e$. 
\end{thm}

\begin{proof}
We can rewrite $\Gamma$ in terms of the eigenbases of $\rho_\s = \sum_{k} \lambda_k^{(\s)} \kb{\s_k}{\s_k}$ and $\rho_\e = \sum_{k} \lambda_k^{(\e)} \kb{\e_k}{\e_k}$: $\Gamma =  \sum_{nmpq} \lambda^{(\e)}_m \lambda^{(\s)}_p |\bra{\s_n \e_m} U \ket{\s_p \e_q}|^2$, where $\sum_{k} \lambda_k^{(\s)}=\sum_{k} \lambda_k^{(\e)}=1$. As the joint $\se$ dimension becomes large, any two bases related by a Haar-random unitary will tend to be mutually unbiased~\cite{BengtssonBruzda2007}; that is, the matrix elements $\bra{\s_n \e_m} U \ket{\s_p \e_q}\rightarrow 1/\sqrt{d_\s d_\e}$. In this limit, we have $\Gamma \rightarrow  \sum_{nmpq} \lambda^{(\e)}_m \lambda^{(\s)}_p /(d_\s d_\e) = 1$.
\end{proof}

In the following Theorem, we show that a fluctuation relation also emerges in the high temperature limit:

\begin{thm}\label{HighTempThm}
As temperature increases, $\Gamma \to 1$ at least linearly with inverse temperature $\beta$.
\end{thm}
\begin{proof}
Consider the completely-positive trace-preserving (CPTP) map $\mathcal{E}_\se : \mathcal{L}(\mathcal{H}_\e) \to \mathcal{L}(\mathcal{H}_\s)$: $\mathcal{E}_\se (\sigma_\e) := \ptr{\e}{U^\dagger {\mathbbm{1}_\s}/{d_\s} \otimes \sigma_\e \, U}$. By the contractivity of the trace distance under CPTP operations, we have  $\tilde{\mu} := \normp{1}{\rho_\e - {\mathbbm{1}_\e}/{d_\e}} \geq \normp{1}{\mathcal{E}_\se \left( \rho_\e - {\mathbbm{1}_\e}/{d_\e} \right)}$. Expanding the action of $\mathcal{E}_\se$ gives
\begin{gather}
\tilde{\mu} \ge \normp{1}{\ptr{\e}{U^\dagger \frac{\mathbbm{1}_\s}{d_\s} \otimes \rho_\e \, U} -  \ptr{\e}{U^\dagger \frac{\mathbbm{1}_\s}{d_\s} \otimes \frac{\mathbbm{1}_\e}{d_\e} U}} = \mu_\s,
\end{gather}
where $\mu_\s = \normp{1}{M_\s - {\mathbbm{1}_\s}/{d_\s}}$. Combining this with Eq.~\eqref{projmax}, we have $d_\s \tilde{\mu} \ge d_\s \mu_\s \ge \mu$.

Next, taking the limit $\beta \to 0$, we have $ \lim_{\beta \to 0} \tilde{\mu} = \lim_{\beta \to 0} |\sum_{k} \left({\expo^{- \beta E_k}}/{Z} - {1}/{d_\e}\right)|$. As $\beta \to 0$, $Z \to d_\e$ and we can expand the exponential. The zeroth order term cancels with the second term in the last equation, giving: $\lim_{\beta \to 0} \tilde{\mu} = ({1}/{d_\e}) |\sum_{k} \sum_{n=1}^\infty (- \beta E_k)^n/{(n!)}|$. Since $\tilde{\mu}$ is behaves linearly with $\beta$ in this limit, it follows that $\Gamma \to 1$ at least linearly.
\end{proof}

\textbf{Speed of convergence.---}A particularly nice feature of our results derived above is that they bound the rate at which the fluctuation relations arise in various regimes. In order to test how fast $\Gamma$ exhibits the results of Theorems~\ref{LargeLimitTheorem}, \ref{HighTotDim}, \& \ref{HighTempThm}, we now explore the statistics of simulated dynamics within the parameter space $(d_\s, d_\e, \beta)$. We construct processes by Haar-randomly sampling unitaries from the joint-space and subsequently define the system and reservoir Hamiltonians
\begin{gather}
H_\s = i \, {\rm tr}_\e[\log (U)] / t 
\quad \mbox{and} \quad
H_\e = i \, {\rm tr}_\s[\log (U)] / t,
\end{gather}
where we choose $t=1$ to fix the units of energy. The notions of high and low temperature depend on the energy level structure of $H_\e$, so we must be careful in comparing processes. At high temperature, we expect significant occupation of all reservoir states, implying $\beta^{-1}\gg |E_N-E_0|$, where $E_N$ and $E_0$ are the highest and lowest reservoir eigenenergies respectively. On the other hand, at low temperature, even the first excited state (with energy $E_1$) has little population, requiring $\beta^{-1}\ll |E_1-E_0|$. Between these two regimes, the temperature is of the same order as the energy splittings in $H_\e$. These considerations motivate the definition of the scaled temperature parameters
\begin{align}\label{eqhighlowtemp}
&\tilde{T}_{\text{low}} := (\beta |E_1 - E_0|)^{-1}, \quad
\tilde{T}_{\text{high}} := (\beta |E_N - E_0|)^{-1}, \notag \\ 
& \quad \mbox{and} \quad \quad \tilde{T}_{\text{mid}} := \frac{N-1}
{\beta \, \sum_{n=1}^{N} |E_{n} - E_{n-1}|},
\end{align}
which are used in the low, high and intermediate temperature regimes respectively.

%Since we examine the effects of temperature and dimensionality on the heat dissipated in random processes applied to nonequilibrium system states, our conclusions describe generic features.

Theorems~\ref{LargeLimitTheorem}, \ref{HighTotDim} and \ref{HighTempThm} manifest themselves in Fig.~\ref{TemperaturePlots}. Plotted in each panel (temperature regime) is a fit of $\mu$ to data from a large number of Haar-randomly sampled interactions for a variety of system and reservoir dimensions. The validity of Theorems~\ref{LargeLimitTheorem} and \ref{HighTotDim} can be seen in any temperature regime: in the low dimensional case (red, dot-dashed), $\mu$ tends to be the largest, with $\mu$ smaller on average for all other cases, where either dimension is large. Reading across the panels of Fig.~\ref{TemperaturePlots} shows that as temperature increases, $\mu \to 0$ independently of $d_\s, \, d_\e$ (note the scale), demonstrating Theorem~\ref{HighTempThm}. Furthermore, in Appendix~\ref{app:boundcomparison} we show that in the cases where the fluctuation relation arises, $\gamma$ almost always provides a tighter bound for the heat than previously known bounds. We now discuss implications and the broader relevance of our findings.

\section{Discussion}

Our ability to coherently control nonequilibrium quantum systems is crucial to developing quantum technologies. Functional quantum technologies must implement irreversible operations, necessarily generating heat which leads to decoherence that negatively impacts performance. In this Article, we have demonstrated the emergence of a fluctuation relation for the heat generated in typical nonequilibrium Landauer processes. The implication is that the heat dissipated into the reservoir in a typical open process is almost always positive. This significantly enhances our understanding of Landauer heat and open evolution of systems in contact with thermal states, as previous studies have been unable to make process-independent statements on the average heat exchanged during nonequilibrium interactions.

Intuition based on the Second Law of Thermodynamics suggests that having the heat typically flow towards the reservoir is a peculiar feature: one might expect that when the reservoir begins at a higher temperature than the system, the reservoir should cool on average. However, due to Theorems~\ref{LargeLimitTheorem}, \ref{HighTotDim} and \ref{HighTempThm}, when the dimension of either subsystem is large or the temperature of the reservoir is high, the heat almost always flows towards the reservoir. Thermodynamic intuition breaks down because it assumes the interaction takes place near equilibrium, where the open evolution cannot perturb a thermal reservoir; in the nonequilibrium setting we consider, this assumption is broken. Randomly sampled unitary processes will generally correspond to Hamiltonians with significant interaction terms that are highly entangling~\cite{Popescu2006}; the local state of the reservoir is therefore almost always more mixed after the interaction than before. The connection between nonequilibrium versions of LP and the Second Law remains unclear and requires further investigation.

Admittedly, real experiments do not have access to random unitary operations, which are often dismissed as ``unphysical''. However, performing a number of interesting tasks efficiently, such as securely erasing quantum information, the decoupling protocol~\cite{Szehr2013,Dupuis2014}, device verification~\cite{Emerson2005,Knill2008} and thermalisation~\cite{Popescu2006,Goldstein2006,Hutter2012}, require operations that mimic sampling from the full space of operations. Indeed, complex physical phenomena are well-approximated by unitary $t$-designs, which agree with the first $t$ moments of the Haar distribution~\cite{Dankert2009,Nakata2017}; especially when there are few particles involved. Here, the concentration of physically relevant random unitary interactions is unlikely to be particularly small. However, our results show that even in these cases, a fluctuation relation arises exponentially quickly. This is somewhat surprising and warrants further study of LP in this microscopic regime, \emph{e.g.}, by relaxing our speed-of-convergence results through a concentration of measure argument over a $t$-design rather than the entire Haar distribution. Moreover, our work will become increasingly important as quantum devices become larger and hotter, since the regimes in which the fluctuation relation arises quickly are exactly those for which the GPM bound typically provides the tightest bound on the heat generated (see Appendix~\ref{app:boundcomparison}). 

On the other hand, our work naturally opens the door to developing tighter bounds on the heat dissipated in a process by including physically motivated constraints within our framework: this might, \emph{e.g.}, restrict the operation space to thermal operations or those generated from Hamiltonians with local interactions~\cite{Goold2015,Lorenzo2015}. The fact that a randomly applied operation will generate heat calls further attention to our need to manage it when building quantum technologies of the future: \emph{How do we best approximate random operations to erase information with minimum heat expenditure?} \emph{How similar does the series of operations in a highly entangling quantum circuit look to a random Landauer process?} Such questions cannot go unanswered if we are to leap into a world run on quantum technologies.

\textbf{Acknowledgements.---}We thank Lucas C{\'e}leri, John Goold and Arul Lakshminarayan for valuable discussions and an anonymous referee for insightful feedback. P. T. is supported by the Australian Government RTP and J. L. William Scholarships. 

\appendix

\section{Proof of Lemma~\ref{lem:levymixed}}
\label{app:levyproof}

The typical application of Levy's Lemma to reduced quantum states is as follows~\cite{Popescu2006}
\begin{lem} \label{lem:levypure}
For any pure state $\phi_\se = U \psi_\se U^\dagger$, where $\psi_\se$ is a fixed system-reservoir pure state and $U$ is a Haar-randomly sampled unitary, and for arbitrary $\epsilon > 0$, the distribution of distances between the reduced density matrix of the system $\phi_s = \ptr{\e}{\phi_\se}$ and the maximally mixed state $\mathbbm{1}/d_\s$ satisfies
\begin{gather}  \label{eq:levyboundpure}
\Prob\left[ \normp{1}{\phi_\s - \frac{\mathbbm{1}_\s}{d_\s}} \geq \sqrt{\frac{d_\s}{d_\e}} + \epsilon \right] \leq2 \exp{ \left( - \frac{d_\s d_\e \epsilon^2 }{16} \right)}.
\end{gather}
\end{lem}

If our states $M_\s$ and $M_\e$ were being generated from pure joint states, we could directly apply Lemma~\ref{lem:levypure} to achieve the desired result in proving Theorem~\ref{LargeLimitTheorem}. However, they are instead being generated from mixed joint states: \emph{i.e.}, $M_\s$ is distributed as $\ptr{\e}{U \tau_\se U^\dagger}$ where $\tau = (\mathbbm{1}_\s/d_\s) \otimes \rho_\e$ (and similarly for $M_\e$). The following argument shows why the standard version of Levy's Lemma does not necessarily hold for the mixed case and motivates our proof of Theorem~\ref{lem:levymixed}.

The initial $\se$ state can always be decomposed as a mixture of pure states: $\tau_\se = \sum_{kl}\lambda_{kl}\ket{kl}\bra{kl}$. We can therefore write
\begin{align}
\mu_s =& \left\|{\rm tr}_\e\left[\sum_{kl}\lambda_{kl}U\ket{kl}\bra{kl} U^\dagger\right] - \mathbbm{1}/d_s\right\|_1.
\end{align}
Defining
\begin{align}
\mu_{kl} :=& \left\|{\rm tr}_\e\left[U\ket{kl}\bra{kl} U^\dagger\right] - \mathbbm{1}/d_s\right\|_1,
\end{align}
we have that $\mu_\s\leq \sum_{kl}\lambda_{kl} \mu_{kl}$ ($\leq \max_{kl} \mu_{kl}$), since the partial trace and the trace norm are both convex functions. While Eq.~\eqref{eq:levyboundpure} would apply to each of the $\mu_{kl}$ with $U$ sampled independently, the upper bound on $\mu_\s$ depends on the full set $\{\mu_{kl}\}$ for each given $U$. Lemma~\ref{lem:levypure} makes no statistical statements about the latter.

Furthermore, since the space of density matrices with fixed spectrum is a geometrically different space from that of pure states (it is a flag manifold rather than a complex projective space), the usual proof of Lemma~\ref{lem:levypure} cannot be trivially modified.
We now proceed to prove our Theorem~\ref{lem:levymixed}.

\begin{proof}
The proof hinges on a version of the well known Lemma by Levy~\cite{LedouxBook}:
\begin{lem}[\emph{\textbf{Levy}}] Consider a manifold $M$ endowed with a metric $g$ and measure $\mu$, and a Lipschitz continuous function $f:M\rightarrow \mathbb{R}$ with Lipschitz constant $\eta$, \textit{i.e.}, $f$ satisfies $|f(x)-f(y)|\leq \eta \|x-y\|_g$ $\forall x,y \in M$. The value of the function is concentrated around its expectation value $\mathbb{E}_x f$ according to the distribution
\begin{gather}
{\rm Prob}\left(f(x)\geq \epsilon + \mathbb{E}_x f\right) \leq 2\alpha_M(\epsilon/\eta),
\end{gather}
where $\alpha_M(x)$ is a concentration function for $M$, defined as (an upper bound on) the measure of the set of points in the space more than a distance $x$ from the minimal-boundary volume enclosing half the space. 
\label{lem:Levy}
\end{lem}

Consider the function $f(U) = \|\sigma_\s(U)-\mathbbm{1}/d_\s\|_1$, the trace distance of the reduced state $\sigma_\s(U) = {\rm tr}_\e[U\tau_\se U^\dagger]$ from the maximally mixed state. Using a reverse triangle inequality and the contractivity of the trace norm under partial trace, we have, for any $U,V\in {\rm SU}(d_{\se})$
\begin{align}
|f(U)-f(V)|=&\left| \left\|\sigma_\s(U)-\frac{\mathbbm{1}}{d_\s}\right\|_1-\left\|\sigma_\s(V)-\frac{\mathbbm{1}}{d_\s}\right\|_1 \right| \nonumber \\ \leq &\, \left\|\sigma_\s(U)-\sigma_\s(V)\right\|_1 \nonumber \\
\leq &\, \left\| U\sigma_{\se}U^\dagger - V\sigma_{\se}V^\dagger \right\|_1\nonumber \\ \leq & 2\left\|U-V\right\|_2, \label{eq:Lipschitz1}
\end{align}
where, for the final inequality, we have used Lemma 1 from Ref.~\cite{EpsteinWhaley2016}, which relates the penultimate quantity to the Hilbert-Schmidt distance ($\|X\|_2=\sqrt{\tr{XX^\dagger}}$) between the two unitaries. Importantly, this distance induces the Haar measure on the group manifold. Eq.~\eqref{eq:Lipschitz1} demonstrates that $f(U)$ is a Lipschitz continuous function on the unitaries with Lipschitz constant $\eta=2$.

Calculating the expectation value of $f(U)$ follows a standard argument~\cite{MMuellerLec6,MMuellerLec8} and is the same as for the case of pure system-reservoir states. The expected trace distance is related to the expected Hilbert-Schmidt distance squared using Jensen's inequality
\begin{align}
\mathbb{E}_U \left\|\sigma_\s(U)-\frac{\mathbbm{1}}{d_\s}\right\|_1 \leq & \sqrt{d_\s} \mathbb{E}_U \left\|\sigma_\s(U)-\frac{\mathbbm{1}}{d_\s}\right\|_2 \nonumber \\
\leq & \sqrt{d_\s}  \sqrt{\mathbb{E}_U\left\|\sigma_\s(U)-\frac{\mathbbm{1}}{d_\s}\right\|_2^2}.\label{eq:expinequality}
\end{align}
The Hilbert-Schmidt distance can then be expanded in terms of the purity of $\sigma_\s$ as
\begin{gather}
\mathbb{E}_U\left\|\sigma_\s(U)-\frac{\mathbbm{1}}{d_\s}\right\|_2^2 = \mathbb{E}_U\tr{\sigma_\s^2}-\frac{1}{d_\s}.\label{eq:expinequality2}
\end{gather}
Lastly, the expectation value of the purity can be calculated for the Haar measure by utilising properties of the swap operator (though the calculation usually involves an average over pure states, it is the same for the unitary orbit of any state); it is
\begin{gather}
\mathbb{E}_U \tr{\sigma_\s^2} = \frac{d_\s+d_\e}{d_\s d_\e + 1},
\end{gather}
which, using Eqs.~\eqref{eq:expinequality}~\&~\eqref{eq:expinequality2}, leads to
\begin{align}
\mathbb{E}_U f \leq & \sqrt{d_\s}\sqrt{\frac{d_\s+d_\e}{d_\s d_\e + 1}-\frac{1}{d_\s}} \nonumber \\
\leq & \sqrt{d_\s}\sqrt{\frac{d_\s+d_\e}{d_\s d_\e}-\frac{1}{d_\s}} = \sqrt{\frac{d_\s}{d_\e}}.
\end{align}

Using Lemma~\ref{lem:Levy}, we have that
\begin{gather}
{\rm Prob}\left(f(U)\geq \epsilon + \sqrt{\frac{d_\s}{d_\e}}\right) \leq 2\alpha_\mathcal{U}(\epsilon/2), \label{eq:prelimineq}
\end{gather}
where $\alpha_\mathcal{U}(x)$ is the concentration function on the group manifold of ${\rm SU}(d_{\se})$ equipped with the Haar measure. We can relate this to the concentration function for the sphere, using a Theorem by Gromov~\cite{Gromov1979,minsurfaces}:
\begin{thm}[\emph{\textbf{Gromov}}]
Let $S^{n+1}(R)$ be the ($n+1$)-sphere of radius $R$, and let $M$ be a closed ($n+1$)-dimensional Riemannian manifold with ${\rm Ric}(M)\geq n/R^2 = {\rm Ric}(S^{n+1}(R))$, where ${\rm Ric}(X)$ is the infimum of diagonal elements of the Ricci curvature tensor on $X$. Choose $M_0\subset M$ to be a domain with smooth boundary and let $B$ be a round ball in $S^{n+1}(R)$ such that
\begin{gather}
\frac{{\rm Vol}(M_0)}{{\rm Vol}(M)}=\frac{{\rm Vol}(B)}{{\rm Vol}(S^{n+1}(R))}.
\end{gather}
It then follows that
\begin{gather}
\frac{{\rm Vol}(\partial M_0)}{{\rm Vol}(M)}\geq \frac{{\rm Vol}(\partial B)}{{\rm Vol}(S^{n+1}(R))}.\label{eq:isoperimetric}
\end{gather}
\label{thm:Gromov}
\end{thm}
That is, if the Ricci curvature is everywhere greater than that of some sphere, then the rate at which volume is enclosed as one moves away from the boundary of a region is at least as great as the corresponding rate for a similar region on the sphere. Thus, if the inequality in Eq.~\eqref{eq:isoperimetric} holds, the corresponding concentration functions are related as $\alpha_M(x)\leq\alpha_{S^{n+1}(R)}(x)$. 

Since the group manifold of ${\rm SU}(d_{\se})$ is compact and simply connected, it has constant, positive Ricci curvature (with respect to the Hilbert-Schmidt distance)~\cite{Milnor1976}. This can be calculated to be ${\rm Ric}(\mathcal{U})=d_{\se}/2$ (see, \emph{e.g.}, Chapter~18 of Ref.~\cite{GallierBook}). The manifold dimension is $d_{\se}^2-1$, so we must compare it with $S^{d_{\se}^2-1}(R)$, finding that the radius must be at least $R_0=\sqrt{2(d_{\se}^2-2)/d_{\se}}$ in order for Theorem~\ref{thm:Gromov} to apply. Choosing the minimal case, we use the Theorem to upper bound $\alpha_\mathcal{U}(x)$ by $\alpha_{S^{d_{\se}^2-1}(R_0)}(x) = \exp(-d_{\se} x^2/4))$~\cite{LedouxBook}. Combining this with Eq.~\eqref{eq:prelimineq} leads to the desired result in Eq.~\eqref{LevyBound}. 

The same argument follows for a function $g(U)=\|\sigma_\e(U)-\mathbbm{1}/d_\e\|_1$; therefore the inequality holds under exchange of system and reservoir labels.
\end{proof}

\section{Bound Comparison}
\label{app:boundcomparison}

Although no discernible hierarchy between $\omega$ and $\gamma$ exists, here we compare the tightness of these bounds to $\beta \braket{\q}$. In Fig.~\ref{BoundComparison}, we sample 5000 Haar-random interactions and initial states to analyse general behaviours in regions of the parameter space. The average heat dissipated in these cases is non-negative, so $\gamma - \omega > 0$ implies that $\gamma$ is a tighter bound.

The only distribution of $\gamma - \omega$ that has a significant proportion of negative data points occurs for small-scale interactions where the process occurs at low temperature (see Fig.~\ref{BoundComparison}~(a)). This shows that $\omega$ can provide a tighter bound to the average heat than $\gamma$ in this regime ($\omega$ outperforms $\gamma$ for $\sim35\%$ of such interactions). However, regardless of the temperature, when either dimension is much larger than the other (or both are large), $\gamma$ almost always provides a tighter bound to the average heat (see Fig.~\ref{BoundComparison}~(b) -- (h)). It is interesting to note that as the dimension of the system increases, $\gamma - \omega$ tends to peak around a specific value (see Fig.~\ref{BoundComparison}~(c), (d), (g) \& (h)).

The tightness of $\gamma$ with respect to the average heat itself can be understood from the distribution of $\beta \braket{\q} - \gamma$. Independent of subsystem dimensions, $\gamma$ provides a tight bound when the process occurs at high temperatures (see Fig.~\ref{BoundComparison}~(b), (d), (f) \& (g)), but a rather poor bound for interactions at low temperatures (where all other known bounds also perform poorly). Note that as the dimension of the reservoir increases, $\beta \braket{\q} - \gamma$ tends to 0 and so the GPM bound is tight.

\begin{figure*}[!h]
\centering
\subfigure{
\includegraphics[width=0.45\linewidth]
{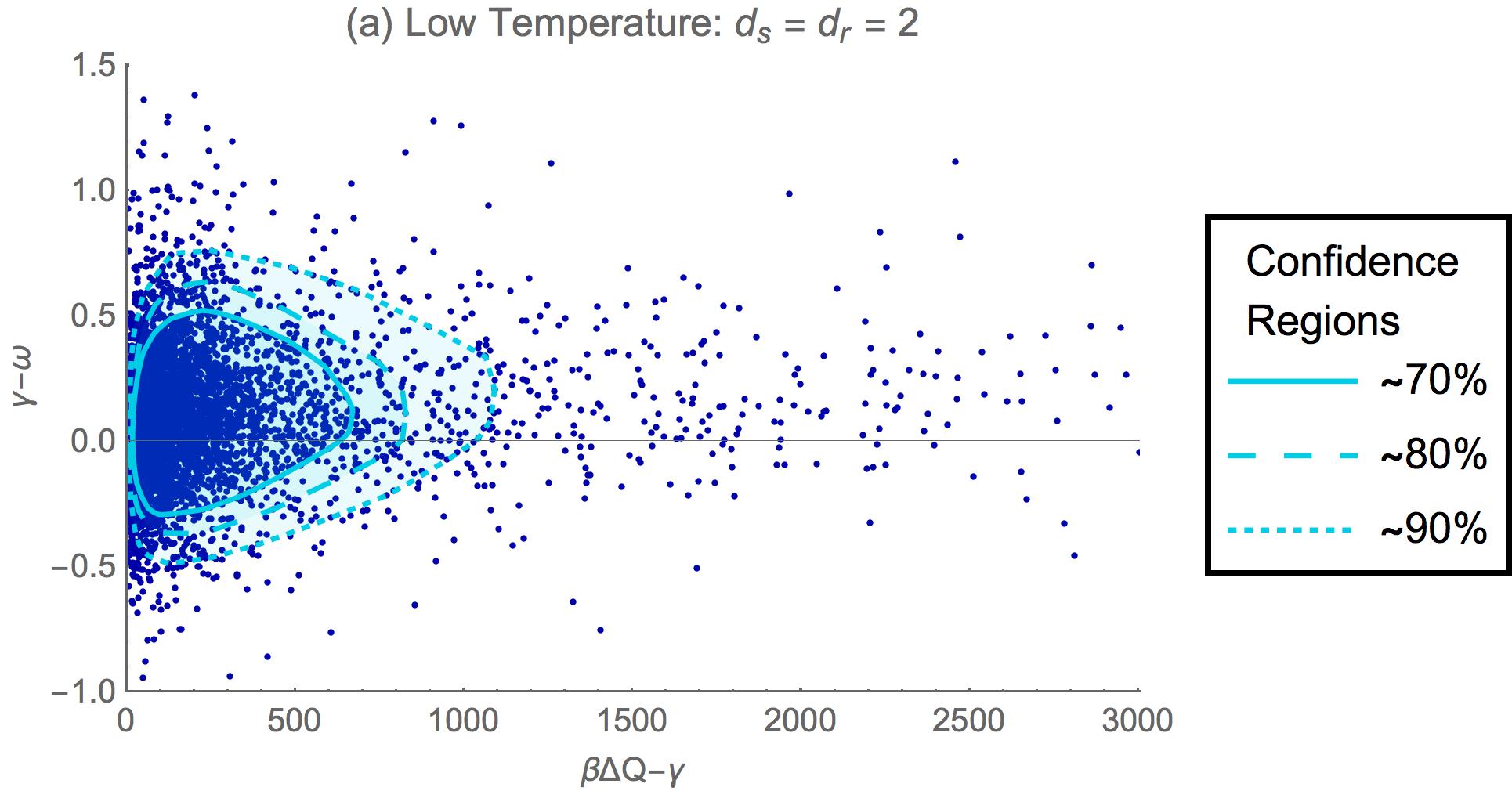}}
%\label{boundcomparison1}
\subfigure{
\includegraphics[width=0.45\linewidth]
{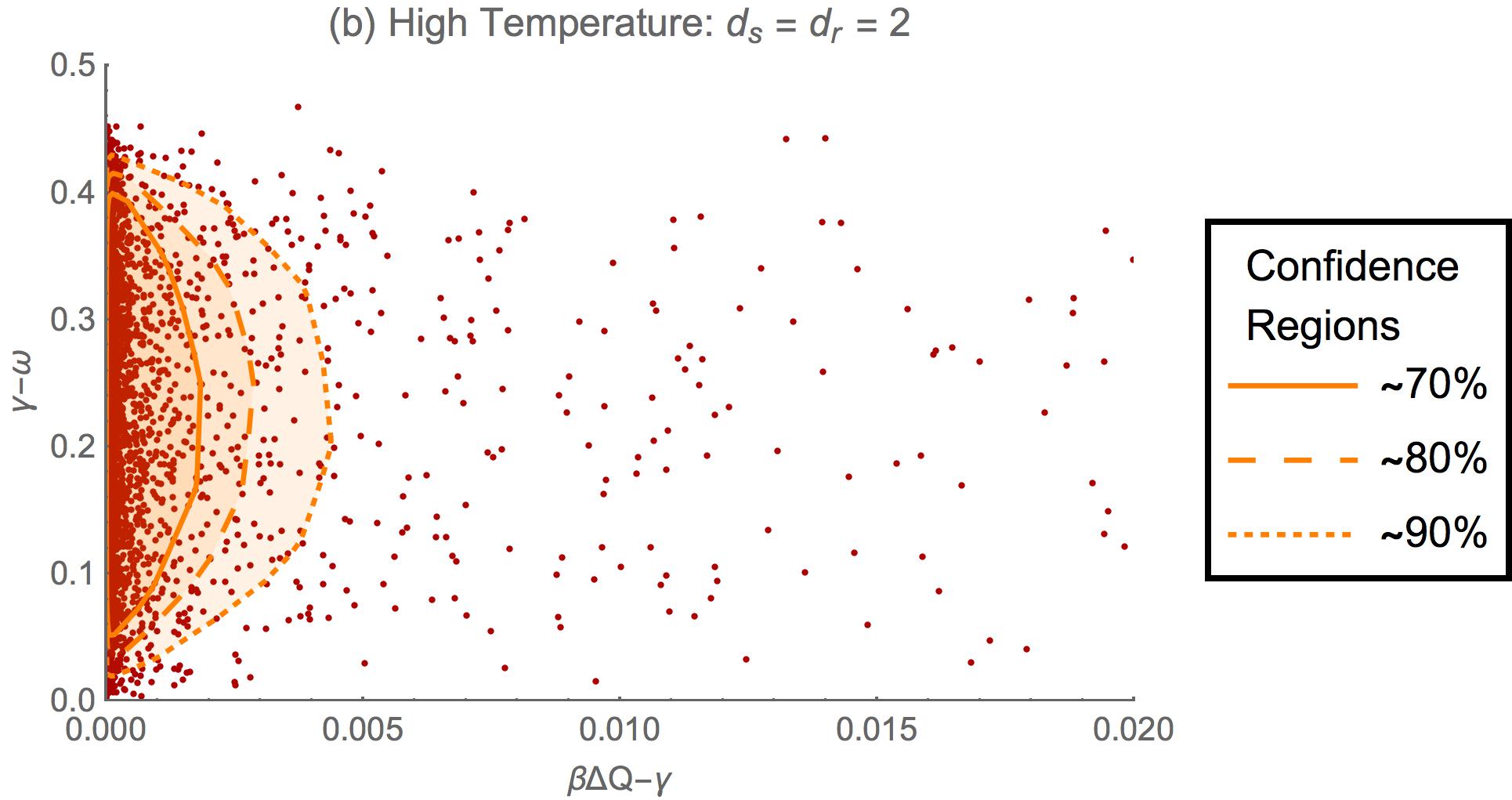}}
%\label{boundcomparison2}
\subfigure{\includegraphics[width=0.45\linewidth]{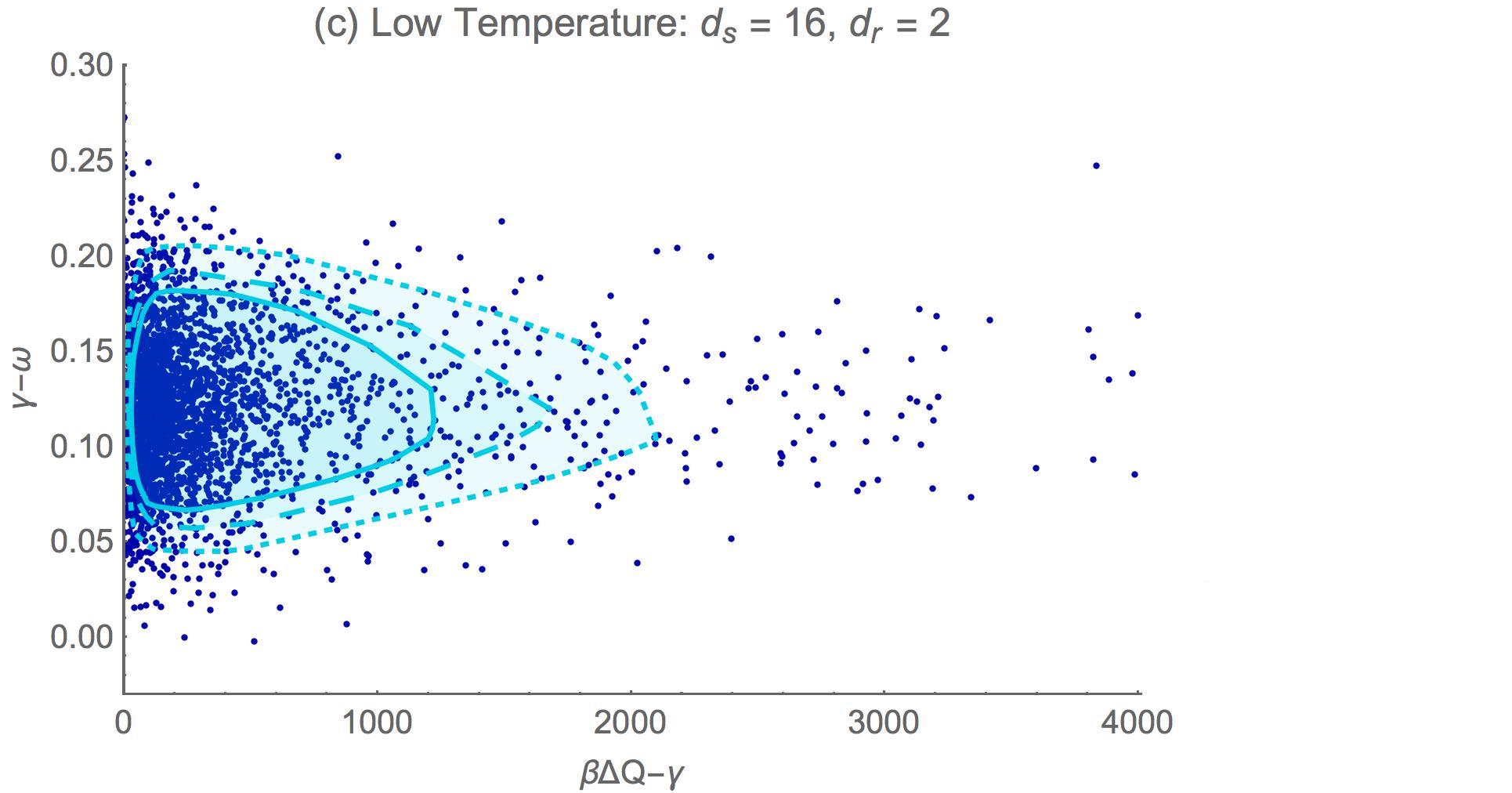}}
%\label{boundcomparison3}
\subfigure{
\includegraphics[width=0.45\linewidth]
{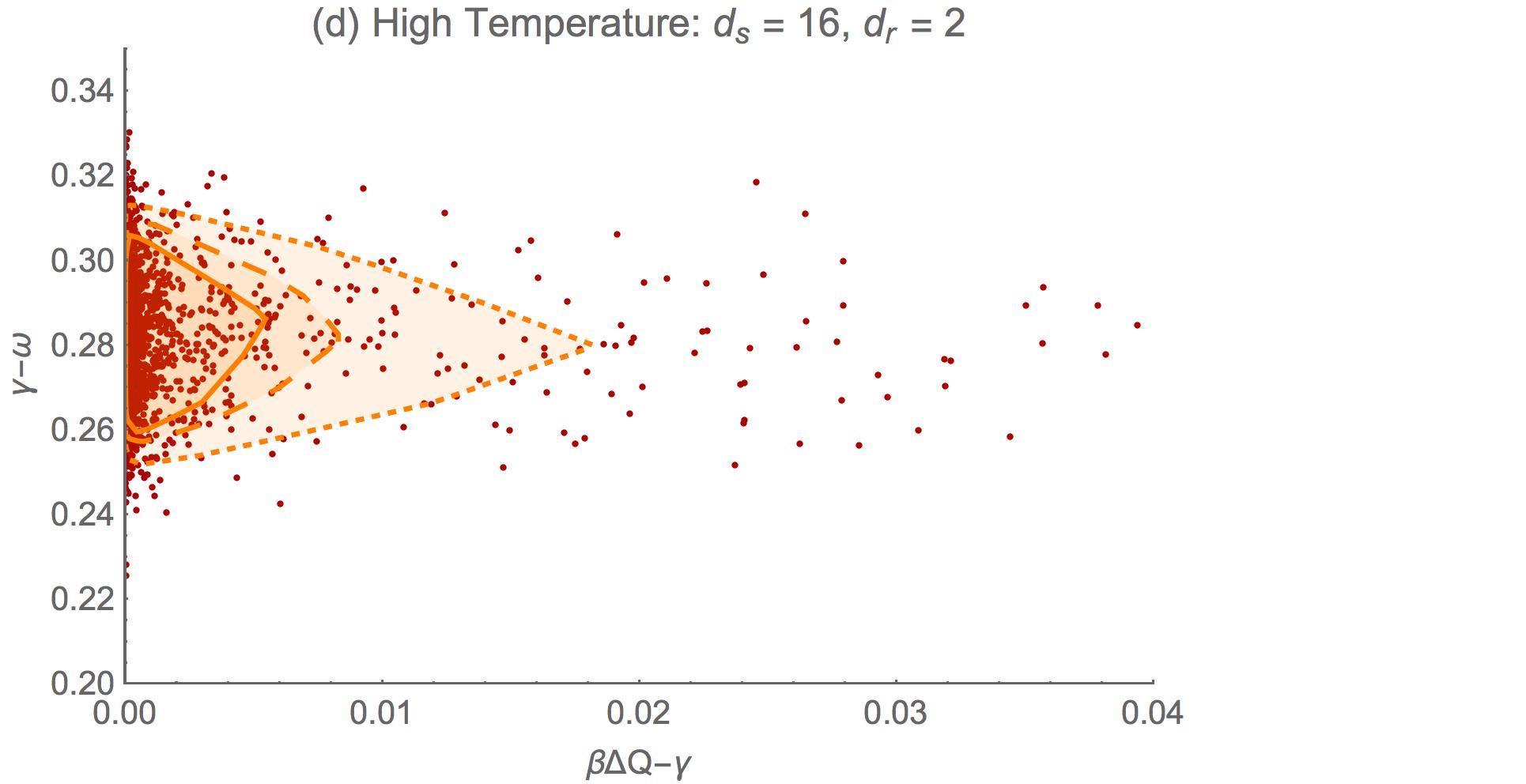}}
%\label{boundcomparison4}
\subfigure{
\includegraphics[width=0.45\linewidth]
{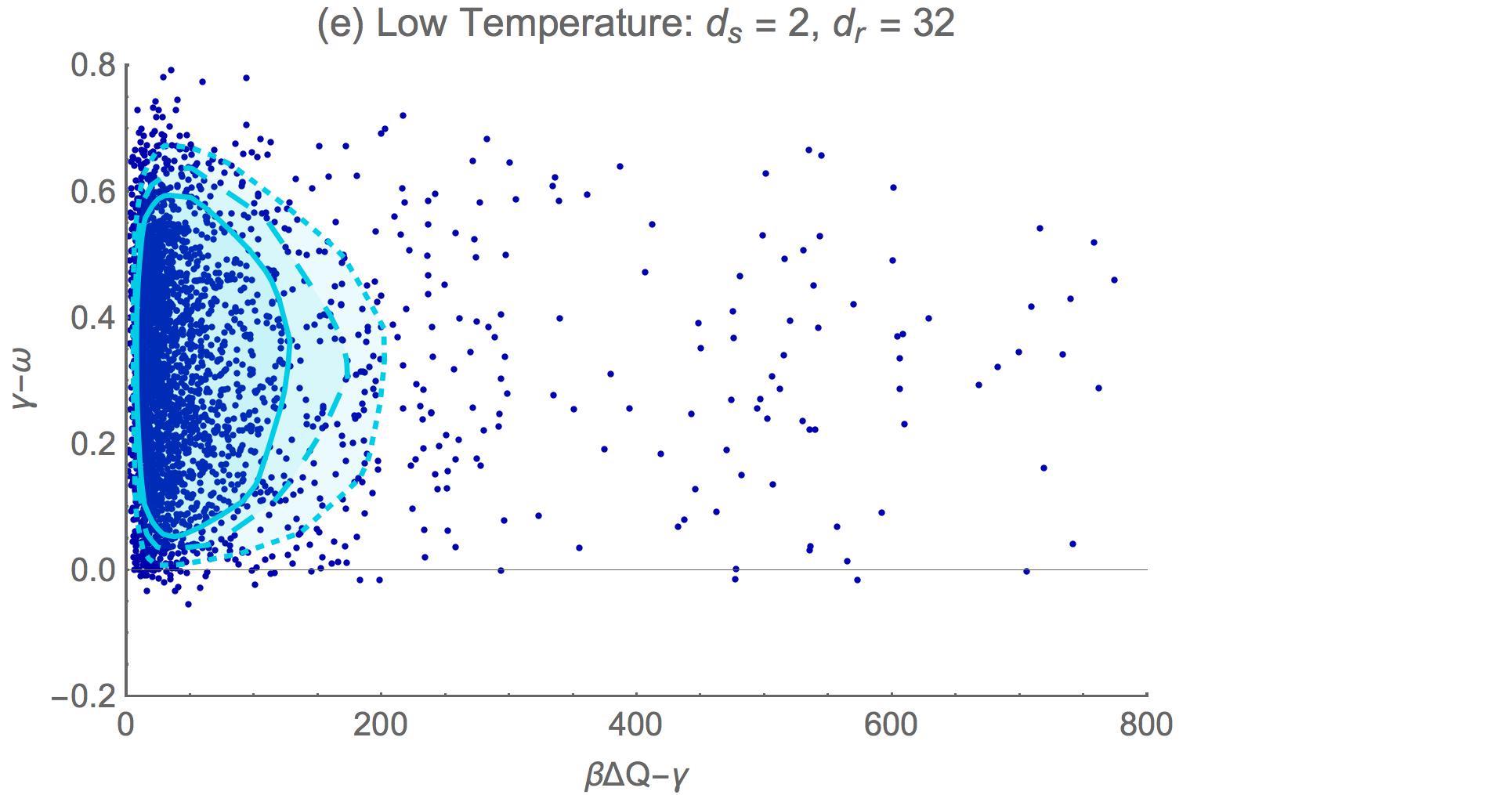}}
%\label{boundcomparison5}
\subfigure{
\includegraphics[width=0.45\linewidth]
{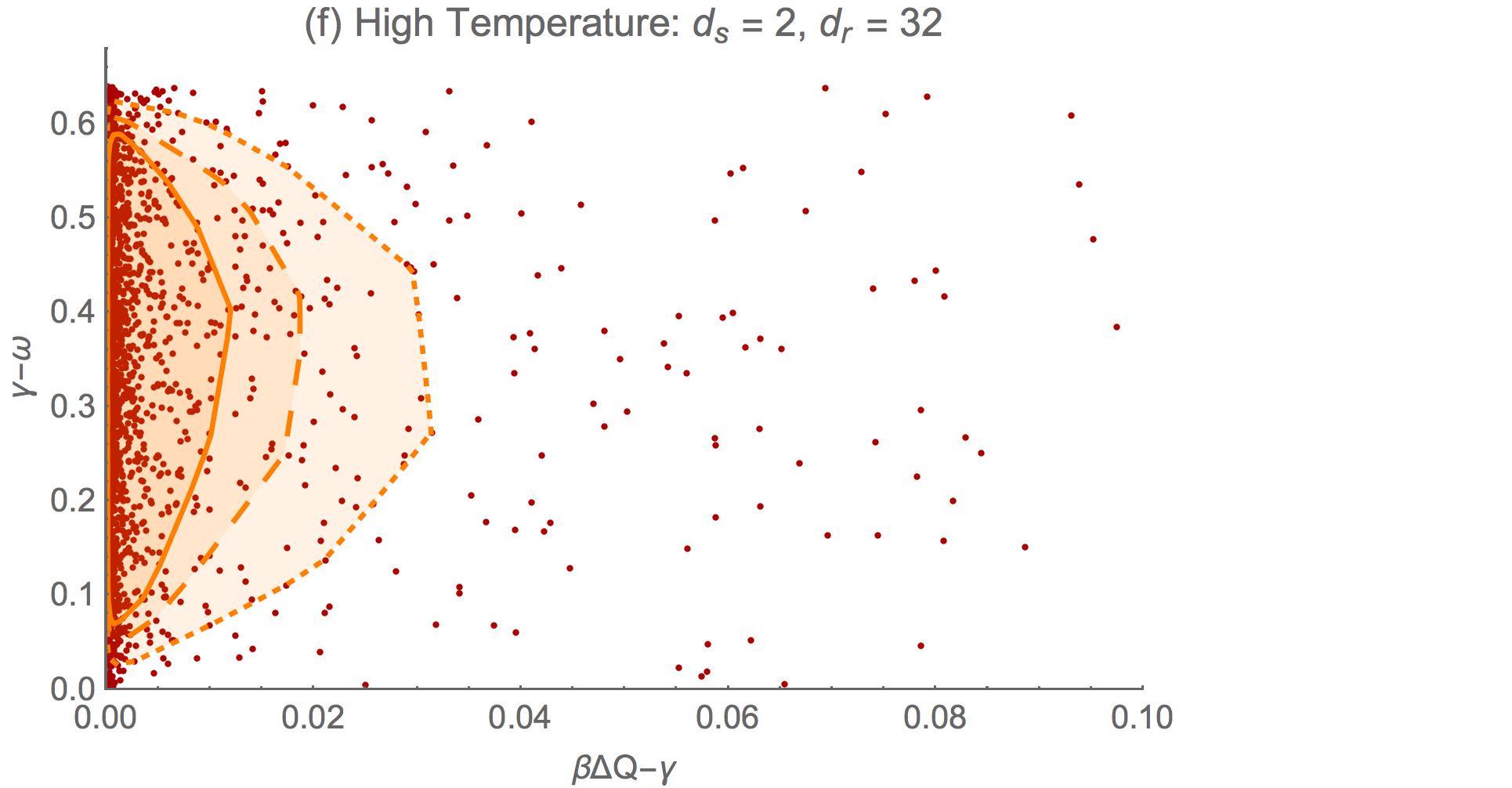}}
%\label{boundcomparison6}
\subfigure{
\includegraphics[width=0.45\linewidth]
{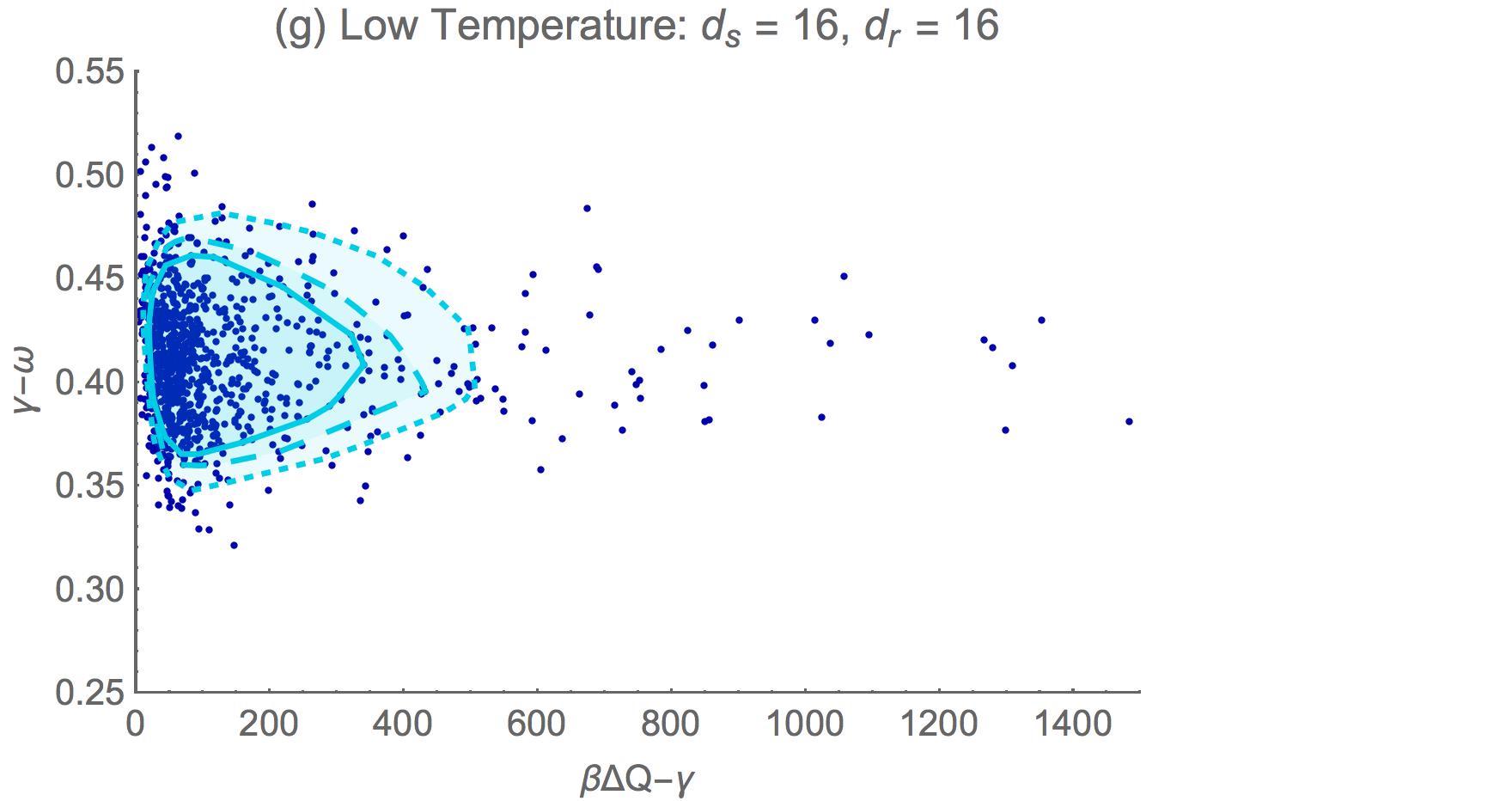}}
%\label{boundcomparison7}
\subfigure{
\includegraphics[width=0.45\linewidth]
{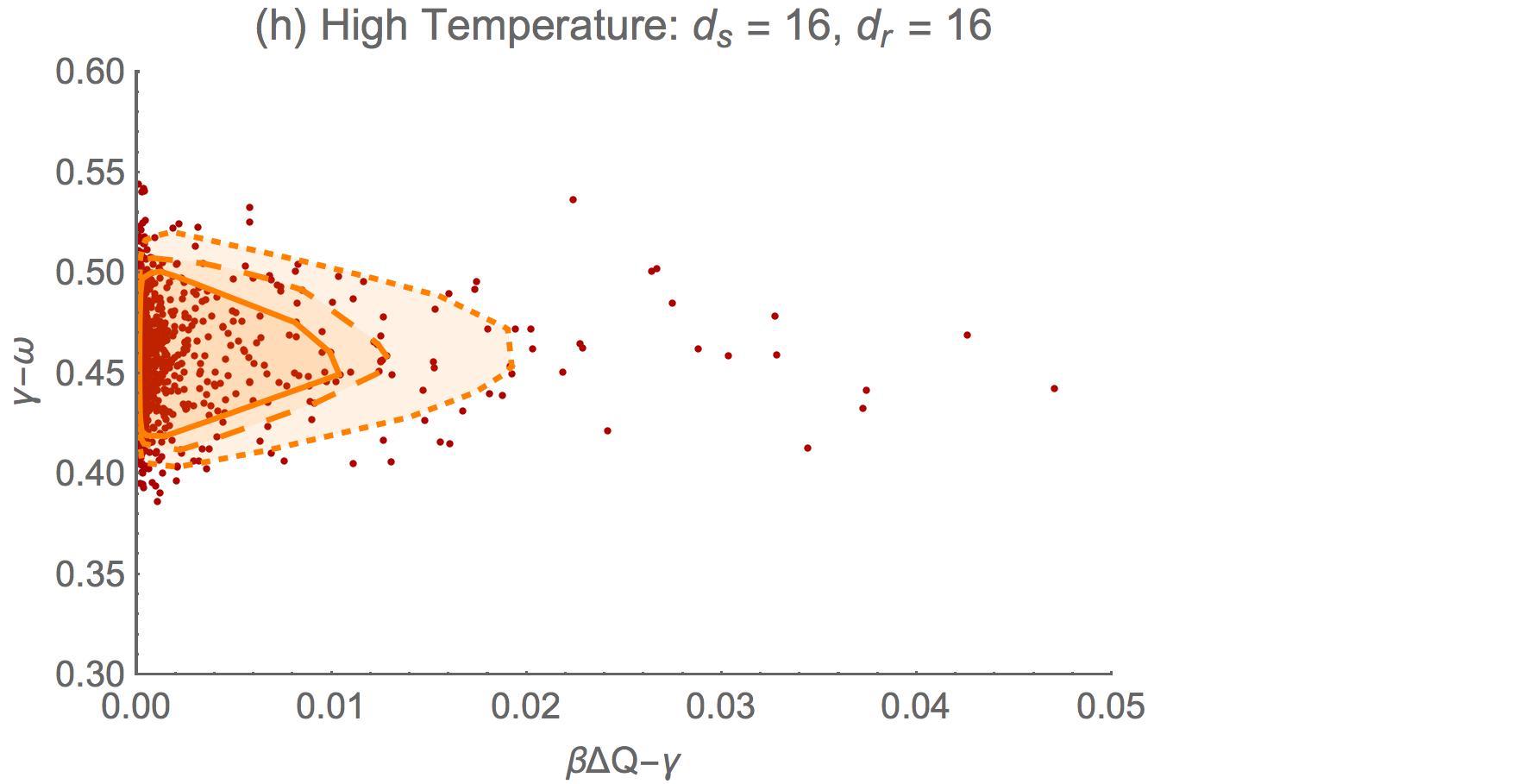}}
%\label{boundcomparison8}
\caption{(Color online). We sample 5000 Haar-random interactions and compute $\beta \braket{\q}$, $\omega$ and $\gamma$ for a variety of dimensions and both low and high temperature regimes. Confidence polytopes containing $\sim70\%$, 80\% and 90\% of the data (solid, dashed and dotted lines respectively) are calculated to highlight general behaviour by peeling off convex hull layers centered on the bivariate median. The only region where $\omega$ significantly outperforms $\gamma$ in tightness to $\beta \braket{\q}$ is the small dimension, low temperature regime (see panel (a), where $\sim35\%$ of interactions generate heat that is closer to $\omega$ than to $\gamma$). In all remaining cases, where fluctuation relation quickly emerges, $\gamma$ almost always provides a tighter bound to the average heat. Furthermore, we see that for any interaction at high temperature, $\gamma$ is a tight bound on $\beta \braket{\q}$. For interactions occurring at low temperatures, $\gamma$ is not a particularly tight bound on the average heat regardless of the dimension (although neither is any previously known bound).} \label{BoundComparison}
\end{figure*}

\clearpage

\bibliography{thermo-fluctuations.bib}

\end{document}